\documentclass{article}

\pdfoutput=1

\usepackage{comment}
\usepackage{PRIMEarxiv}
\usepackage{tikz}
\usepackage{soul}
\usepackage[utf8]{inputenc} 
\usepackage[T1]{fontenc}    
\usepackage{hyperref}       
\usepackage{url}            
\usepackage{booktabs}       
\usepackage{amsfonts}       
\usepackage{nicefrac}       
\usepackage{microtype}      
\usepackage{lipsum}
\usepackage{amsmath}
\usepackage{amssymb}
\usepackage{amsthm}

\newtheorem{theorem}{Theorem}
\newtheorem{lemma}{Lemma}

\theoremstyle{definition}

\usepackage{subcaption}
\usepackage{fancyhdr}       
\usepackage{graphicx}       
\graphicspath{{media/}}     
\newcommand{\cmt}[1]{}
\title{Optimal Pricing Schemes for Identical Items with Time-Sensitive Buyers}

\author{
  Zhengyang Liu \\
  Beijing Institute of Technology\\
  \texttt{zhengyang@bit.edu.cn} \\
   \And
  Liang Shan\\
  Renmin University of China\\
 \texttt{shanliang@ruc.edu.cn}
 \And
  Zihe Wang \thanks{Zihe Wang is the corresponding author}\\
  Renmin University of China\\
 \texttt{wang.zihe@ruc.edu.cn}
}

\begin{document}
\maketitle

\begin{abstract}
Time or money? That is a question! In this paper, we consider this dilemma in the pricing regime, in which we try to find the optimal pricing scheme for identical items with heterogenous time-sensitive buyers. We characterize the revenue-optimal solution and propose an efficient algorithm to find it in a Bayesian setting. Our results also demonstrate the tight ratio between the value of wasted time and the seller's revenue, as well as that of two common-used pricing schemes, the $k$-step function and the fixed pricing. To explore the nature of the optimal scheme in the general setting, we present the closed forms over the product distribution and show by examples that positive correlation between the valuation of the item and the cost per unit time could help increase revenue. To the best of our knowledge, it is the first step towards understanding the impact of the time factor as a part of the buyer cost in pricing problems, in the computational view.
\end{abstract}

\section{Introduction}
The Time~vs.~Money dilemma exists for everyone in our life. When buying durable, more or less expensive, one may wait for the discount activities or choose the lowest price by visiting various stores. However, the valuation of time differs, the higher buyer's valuation of time is, the less time she wants to waste on the deal. Another scenario in the real world is the online in-app purchase. In the App Store or Google Play, many applications charge in their apps to promote the efficiency of the service or the user experience. They often prepare other ways for users to implement similar functionalities or finish the same jobs by putting more effort. For example, a platform could set different limits on the speed of data transmission.  Once we can distinguish the CPUTs (Cost Per Unit Time) of various agents, one can design more reasonable pricing schemes, that is ``lose the money, or lose the time''.

From the sellers' perspective, they can leverage the buyers' sensitiveness for time or money to gain more revenue. We illustrate the ``Double 11'', the biggest online shopping festival in the world, as an example. The shopping festival has already become the world's largest shopping festival since 2015~\cite{Chen:2016tz}. During the shopping period, the online shopping platform tends to keep customers spending as long time as possible by designing complicated schemes requiring customers to spend a lot of effort to get the discount vouchers~\cite{Chen:2021va}. The case makes the customer trapped by the time vs. money dilemma. They either devote time effort following the rules of the platform to get the discount voucher or take the item directly, without any discount.

In this paper, we try to understand the customers' trade-off between time and money in the view of sellers. We consider the optimal dynamic pricing scheme for heterogeneous unit-demand customers.
In our setting, a single seller tries to sell identical items with unlimited supply to unit-demand buyers. To attract more potential customers and attain more revenue, the seller designs many complicated promotions. The rules of these promotions are not easy to understand or you need wait for the starting time of the promotions, so one needs spending time on the vouchers. For convenience, we assume that if a buyer spends time $t$ on her purchase, the price should be deterministic, say $p(t)$, where $p:\mathbb{R}\to\mathbb{R}$ is the pricing function set by the seller. The objective of the seller is to come up with a pricing scheme seeking to maximize the revenue.
\cmt{
The model can be simplified as the seller designs a promotion function $p(t)$ such that if you spend $t$ time in the participate the promotion activity, you will get the item at price $p(t)$.
}
Each buyer has her own CPUT (Cost Per Unit Time), denoted by $\theta$, and the valuation to the item, denoted by $v$.
A buyer tends to buy the item if the condition $\min_t \{\theta t+p(t)\}\le v$ is satisfied. The seller's revenue is $p(t)$, where $t$ is the time spent by the buyer.



In this paper, we introduce the pricing scheme problem with time-sensitive buyers and make a few contributions: 
\begin{itemize}
    \item We formulate the pricing scheme problem in our setting as an optimization problem in a Bayesian way. We then propose a polynomial algorithm to find the optimal scheme in the discrete distribution setting. To handle the continuous distribution setting, we can discretize the distributions and reduce to the problem with the discrete distribution. We give an upper bound of the revenue loss caused by this discretization. We also 
    show that the calculus method is not applicable if the continuous distribution is not good enough.
    
    \item We consider the revenue change and measure the value of wasted time in our pricing scheme. We show a tight ratio between the performance of the optimal $k$-step pricing scheme and the fixed pricing scheme in terms of the revenue. We show the value of wasted time could far outweigh the seller's revenue and provide a tight ratio between the time loss and the optimal revenue.
    
    \item We investigate a special case to gain more intuitions. \cmt{We prove that the well-known posted-pricing scheme (e.g., the price tag in the supermarket) is optimal when the valuation and the CPUT are independently distributed.} Optimal dynamic pricing schemes is more likely to show advantages when the valuation and CPUT are positive correlated.
    
\end{itemize}
\subsection{Related Works}
\cite{Coase:1972we} initiates the dynamic pricing with intertemporal demand and strategic buyers. \cite{Coase:1972we}  conjectures that the price converges to the cost when selling the durable goods to the strategic buyers if not committing to a posted-price. Later, many works~\cite{Su:2007wu,Araman:2009ty,Cachon:2011tw,Yu:2016vj,Golrezaei:2020td} demonstrate that dynamic pricing can gain more revenue for the seller when we have no idea of the total demand.
The most related work is \cite{yang2021costly}. He considers a general model where a principal screens an agent with multidimensional private information: a one-dimensional productive component and a multidimensional costly component. In our model, a buyer’s value for an item corresponds to the productive component and a buyer’s cost per unit time (CPUT) corresponds to the costly component. We consider a special case of Yang’s model from a more computational perspective. His result implies that if the productive component (value for the item) and the costly component (CPUT) have positive correlation (negative correlation in our model), then there exists an optimal mechanism that involves no costly screening (posted pricing in our model). We complement his result by solving a positive correlation case. 
\paragraph{Dynamic Scheme Design.}
One related line of research to our paper is using the dynamic approach to design the pricing scheme~\cite{Gallien:2006vg,Bergemann:2010ul,Akan:2015wh,Golrezaei:2017wq,Kakade:2013ud,Pavan:2014ti}. A seminal paper by~\cite{Gallien:2006vg} considers an infinite horizon and discounted sets and designs the first tractable pricing policy for strategic buyers. \cite{Board:2016uj} discuss the same problem with discrete and finite time horizons. Both works assume that buyers' valuations could be discounted with some
parameters geometrically, and the discounted knowledge is even public to the seller. We consider the same finite horizon setting  as~\cite{Pai:2013tf,Board:2016uj}, which is more relevant to the applications in daily life. In addition, we have no restriction to the discounted value of buyers, except that they are non-increasing over time. One of the most related papers to our work is that by~\cite{Chen:2018tx}, who propose a $0.29$-approximation revenue optimal pricing algorithm in the sense that valuations of the customers decay with different rates, and the customers arrive with a Poisson process.


All the works in this research line consider a different buyer model and assume the demands and supplies change over time. In contrast, buyers in our paper are unit-demand, and the supply is unlimited. We face the issue of wasted time while they do not.

\paragraph{Bayesian Revenue Maximization.}
Since the celebrated work due by~\cite{Myerson:1981wi}, the revenue
both computer science and economic communities has extensively studied maximization in Bayesian setting~\cite{Bei:2011uw,Cai:2013vt,Haghpanah:2015vq,Bei:2017ty,Cai:2019vq,Zihe:2022,tang2016optimal,wang2015optimal}.
One of the prominent factors that contrast our paper with the aforementioned research is introducing the time~vs.~money dilemma in the views of buyers. All related papers in the literature assume that buyers are ideally rational, meaning that they can come up with the best way to maximize their utility without losing time.
\cmt{To the best of our knowledge, it is the first time to consider the time factor as a part of the cost in this literature.
The hardness results of variant pricing scheme are extensively studied in TCS community. We conjecture that finding an optimal pricing scheme in our work is also intractable.}

\section{Preliminaries} 
In our model, a seller sells identical and indivisible items with unlimited supply to unit-demand buyers, who needs at most one item. Each buyer has a type, consists of two parameters $(\theta,v)$, where $\theta$ represents her CPUT (Cost Per Unit Time), and $v$ is her valuation to the item. We assume that buyers' types are generated by a joint probability distribution $F$ over $\theta$ and $v$. For convenience, we also denote by $f$ the probability density function of $F$. The goal of the problem is to maximize the seller's expected revenue with respect to the joint distribution $F$.

To maximize the revenue, the seller can design a publicly known pricing function $p:\mathbb{R}\to \mathbb{R}$, which only depends on the time the buyer spends on the deal. That is, any buyer needs to pay the same price $p(t)$ if she spends time $t\ge0$. W.L.O.G., we assume that the price charged by the seller cannot increase if one spends more time on the deal, i.e., the function $p(t)$ is non-increasing in terms of the time $t$. A buyer's action is {\em to take the item at the right time or leave it}. According to the pricing function, each buyer makes a trade-off between the time she spends and the price she pays. We assume each buyer is strategic. Therefore she tries to get more utility by spending the right amount of time in promotions. We also assume the buyer is individually rational (IR) such that she always gets a non-negative utility. Hence the buyer with type $(\theta,v)$ will buy the item if and only if $\min_t \{p(t)+\theta t\}\le v$, that is, the total cost including money and time is at most the value to an item in the view of the buyer.
Once the buyer with type $(\theta,v)$ buys the item, the set of best actions of the buyer is 
$T^*(\theta)=\arg\min_t \{p(t)+\theta t\}$ 
which only depends on CPUT $\theta$. 
Let $t^*(\theta)=\min\left\{t\mid t\in T^*(\theta)\right\}$ denote the best action of the buyer, that is, when multiple best actions exist, we assume the buyer will purchase the item with the least time on the deal. 
We can formulate such a problem for the seller as follows,
\begin{align}
\max_{p(\cdot)} & \int\limits_{(\theta,v)\sim F} p(t^{*}(\theta)) \mathbb{I}\left\{ \theta
  t^{*}(\theta)+p(t^{*}(\theta))\le v \right\} f(\theta,v) 
  \mathrm{d}\theta \mathrm{d} v,\nonumber\\
 \text{s.t.} & ~~ t^*(\theta)=\min\left\{t\mid t\in T^*(\theta)\right\},\label{eq:main}
\end{align}
where $\mathbb{I}\{\cdot\}$ is an indicator function.


\section{Characterizations of the Optimal Pricing Scheme}
Before designing our solution for the problem, we first characterize several properties for the pricing function. The first one is the convexity of the region where buyers cannot afford no matter how much time they spend.
\begin{lemma}
Given the pricing function $p:\mathbb{R}\to\mathbb{R}$, the set of types $\left\{(\theta,v)\mid \min_t \{p(t)+\theta t\}> v\right\}$ is convex.
\label{lemma_convex}
\end{lemma}

\begin{proof}
Given the time $t$ spent on the deal, we denote by $A(t)$ the set of types where the buyer spends time $t$ but receive the negative utility, formally,
\begin{equation*}
    A(t)=\{(\theta,v)\mid p(t)+\theta t>v\}.
\end{equation*}
Note that a buyer would not like to buy the item if the utility is always negative, no matter how much time she spends on the deal. Therefore, the buyer gets no item if and only if her type belongs to the intersection set $\bigcap_{t\ge 0} A(t)$. One can verify that the set $A(t)$ is the half-plane, as shown in Figure~\ref{fig:1.1}. By the fact that the intersection of convex sets is also convex, we know that $\bigcap_{t\ge 0}A(t)$ is also convex.
\end{proof}

\begin{figure}
    \centering
    \begin{subfigure}[t]{4cm}
        \centering
\tikzset{every picture/.style={line width=0.75pt}} 
\resizebox{4cm}{3cm}{%
\begin{tikzpicture}[x=0.75pt,y=0.75pt,yscale=-1,xscale=1]

\draw  (98.5,270.25) -- (308,270.25)(120.5,105) -- (120.5,293.25) (301,265.25) -- (308,270.25) -- (301,275.25) (115.5,112) -- (120.5,105) -- (125.5,112)  ;
\draw    (97.5,160.75) -- (319.5,160.75) ;
\draw    (102,247.75) -- (269.5,110.75) ;
\draw    (99.5,220.75) -- (344.44,117.28) ;
\draw  [draw opacity=0][fill={rgb, 255:red, 235; green, 235; blue, 235 }  ,fill opacity=1 ] (241.58,161.29) -- (173.25,190.13) -- (120.71,233.07) -- (120.92,269.63) -- (290,269.71) -- (290,161.21) -- cycle ;
\draw [line width=0.75]    (261.6,117.6) -- (269.71,127.27) ;
\draw [shift={(271,128.8)}, rotate = 229.99] [color={rgb, 255:red, 0; green, 0; blue, 0 }  ][line width=0.75]    (4.37,-1.32) .. controls (2.78,-0.56) and (1.32,-0.12) .. (0,0) .. controls (1.32,0.12) and (2.78,0.56) .. (4.37,1.32)   ;
\draw    (318,128) -- (325.96,141.04) ;
\draw [shift={(327,142.75)}, rotate = 238.61] [color={rgb, 255:red, 0; green, 0; blue, 0 }  ][line width=0.75]    (4.37,-1.32) .. controls (2.78,-0.56) and (1.32,-0.12) .. (0,0) .. controls (1.32,0.12) and (2.78,0.56) .. (4.37,1.32)   ;
\draw    (309.5,160) -- (309.5,174.75) ;
\draw [shift={(309.5,176.75)}, rotate = 270] [color={rgb, 255:red, 0; green, 0; blue, 0 }  ][line width=0.75]    (4.37,-1.32) .. controls (2.78,-0.56) and (1.32,-0.12) .. (0,0) .. controls (1.32,0.12) and (2.78,0.56) .. (4.37,1.32)   ;

\draw (104.4,114) node [anchor=north west][inner sep=0.75pt]    {$v$};
\draw (272.4,106.6) node [anchor=north west][inner sep=0.75pt]  [font=\footnotesize]  {$A( 10)$};
\draw (331.4,127.6) node [anchor=north west][inner sep=0.75pt]  [font=\footnotesize]  {$A( 5)$};
\draw (316.4,163.6) node [anchor=north west][inner sep=0.75pt]  [font=\footnotesize]  {$A( 0)$};
\draw (285,273.4) node [anchor=north west][inner sep=0.75pt]    {$\theta $};
\draw (183,204.9) node [anchor=north west][inner sep=0.75pt]    {$\bigcap _{t\geq 0} A( t)$};

\end{tikzpicture}
}
        \caption{the concave and weakly increasing separation function}
        \label{fig:1.1}
    \end{subfigure}
    \hspace{1in}
    \begin{subfigure}[t]{4cm}
        \centering
        
\resizebox{4cm}{3cm}{%

\tikzset{every picture/.style={line width=0.75pt}} 

\begin{tikzpicture}[x=0.75pt,y=0.75pt,yscale=-1,xscale=1]

\draw  (98.5,270.25) -- (308,270.25)(120.5,105) -- (120.5,293.25) (301,265.25) -- (308,270.25) -- (301,275.25) (115.5,112) -- (120.5,105) -- (125.5,112)  ;
\draw    (120.56,220.19) .. controls (155.06,155.69) and (243.8,131.2) .. (269.89,130.22) ;
\draw    (101.4,203.87) -- (272.2,110.62) ;
\draw   (119.2,193.5) .. controls (119.2,193) and (119.6,192.6) .. (120.1,192.6) .. controls (120.6,192.6) and (121,193) .. (121,193.5) .. controls (121,194) and (120.6,194.4) .. (120.1,194.4) .. controls (119.6,194.4) and (119.2,194) .. (119.2,193.5) -- cycle ;
\draw   (186.67,156.57) .. controls (186.67,156.07) and (187.07,155.67) .. (187.57,155.67) .. controls (188.06,155.67) and (188.47,156.07) .. (188.47,156.57) .. controls (188.47,157.06) and (188.06,157.47) .. (187.57,157.47) .. controls (187.07,157.47) and (186.67,157.06) .. (186.67,156.57) -- cycle ;

\draw (104.4,114) node [anchor=north west][inner sep=0.75pt]    {$v$};
\draw (285,273.4) node [anchor=north west][inner sep=0.75pt]    {$\theta $};
\draw (69.9,175.7) node [anchor=north west][inner sep=0.75pt]  [font=\footnotesize]  {$p\left( t^{*}( \theta )\right)$};
\draw (132.57,129.87) node [anchor=north west][inner sep=0.75pt]  [font=\footnotesize]  {$\ell '( \theta ) =t^{*}( \theta )$};
\draw (262.07,131.87) node [anchor=north west][inner sep=0.75pt]  [font=\footnotesize]  {$\ell ( \theta )$};

\end{tikzpicture}

}
        \caption{the pricing function and the separation function}
        \label{fig:1.2}
    \end{subfigure}
    \caption{An illustration for the separation function}
    \label{fig:1}
\end{figure}
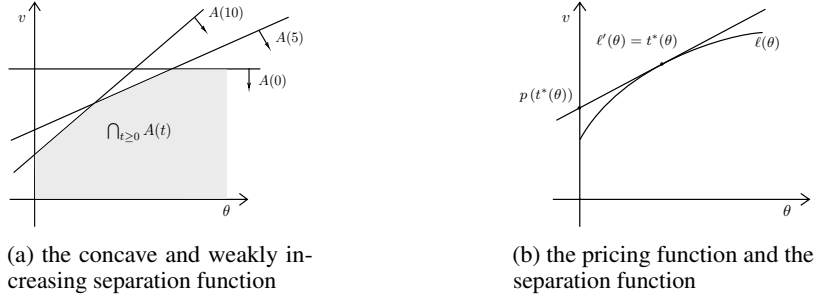

Given a pricing function $p(\cdot)$, we have such a convex set as mentioned above.
We call the upper boundary of the convex set as the {\em separation function}, denoted by $\ell_p:\mathbb{R}\to\mathbb{R}$. According to Lemma~\ref{lemma_convex},
the function $\ell_p(\cdot)$ is weakly increasing and concave. Therefore, it is differentiable at all but few points. In other words, the derivative of $\ell(\cdot)$ is well defined on almost every point.
The following lemma demonstrates that for each type $(\theta, \ell_p(\theta))$ on this function, the corresponding buyer will spend time $t$ which exactly equals to its derivative of the function at $\theta$.
\begin{lemma}
\label{lemma_time}
Given the pricing function $p$, if $t^*(\cdot)=\min\{t\mid t\in T^*(\cdot)\}$ is continuous at $\theta$, we have
$\ell_p'(\theta)=t^*(\theta)$. 
\end{lemma}
\begin{proof}
According to the definition of separation function $\ell_p$ and the fact that the pricing function is non-increasing, given a sufficiently small quantity $\delta>0$, we have
\begin{align*} 
    \ell_p(\theta)&=\theta t^*(\theta)+p(t^*(\theta)),\\
    \ell_p(\theta+\delta) &\le (\theta+\delta) t^*(\theta)+p(t^*(\theta)),
\end{align*}
    and that is
\begin{equation}\label{eq:e1}
    \ell_p(\theta+\delta)-\ell_p(\theta)\le \delta  t^*(\theta).    
\end{equation} 
Similarly, we have the followings,
\begin{align*}
    \ell_p(\theta)&\le \theta t^*(\theta+\delta)+p(t^*(\theta+\delta)),\\
    \ell_p(\theta+\delta)&=(\theta+\delta) t^*(\theta+\delta)+p(t^*(\theta+\delta )),
\end{align*}
    and that is
\begin{equation}\label{eq:e2}
    \ell_p(\theta+\delta)-\ell_p(\theta)\ge \delta  t^*(\theta+\delta).
\end{equation} 
By combining Equations~\eqref{eq:e1} and~\eqref{eq:e2}, we can bound $\ell_p(\theta+\delta)-\ell_p(\theta)$ from below and above:
\begin{align*}
    \delta t^*(\theta+\delta)\le& \ell_p(\theta+\delta)-\ell_p(\theta)\le \delta  t^*(\theta).
\end{align*}
Dividing the above inequality by $\delta$ and taking the limit as $\delta$ approaches zero yields the constraint
\begin{equation*}
    \ell_p'(\theta)=t^*(\theta).
    \qedhere
\end{equation*}
\end{proof}

The payment by the buyer with type $\theta$ is $p(t^*(\theta))$, shown in Figure~\ref{fig:1.2}. Note that for a buyer with the type $(\theta,\ell_p(\theta))$ which is on the separation function, her utility is zero, that is
\begin{align}
   p(t^*(\theta))=\ell_p(\theta)-\theta t^*(\theta) = \ell_p(\theta) - \theta\ell_p'(\theta). \label{infer_p}
\end{align}

Lemma~\ref{lemma_convex} states that, given a pricing function $p(\cdot)$, our separation function $\ell_p(\cdot)$ is determined, concave and weakly increasing. In reverse, given a concave and weakly increasing separation function $\ell_p(\cdot)$, we can also infer the function $p(\cdot)$. We define $H$ as the set of derivatives $\ell_p'(\theta)$, i.e., $H=\{\ell_p'(\theta)\mid \forall \theta\in\mathbb{R}\}$. We define $p(\ell_p'(\theta))=\ell_p(\theta)-\theta\ell_p'(\theta)$ for all $\theta$ and thus the value of function $p$ is defined on points in $H$. Even though there might exist multiple $\theta$'s with the same value of $\ell_p(\theta)$, at that time, the values of $\ell_p(\theta)-\theta\ell_p'(\theta)$ are identical. As above, the function $p$ is well defined. Furthermore, it is also possible that different pricing functions might result in the same function $\ell_p$. Given the separation line function $\ell_p(\cdot)$, the differences in these preimages are the values on points excluding the set $H$. One of the preimages could be $p(t)=p\left(\max\{s\mid s\le t \text{ and } s\in H\}\right)$.

As discussed above, we could design the separation function $\ell_p$, instead of designing the pricing function $p$ directly. In Equation~\eqref{eq:main}, the integral of $v$ can be simplified as follows,
\begin{eqnarray*}
\int_{v} \mathbb{I}\left\{ \theta t^{*}(\theta)+p(t^{*}(\theta))\le v  \right\}f(\theta,v) \mathrm{d}v &=&\int_{v\ge \ell_p(\theta)}f(\theta,v) \mathrm{d}v\\
&=&(1-F_{v|\theta}[\ell_p(\theta)|\theta])f_{\theta}(\theta).
\end{eqnarray*}
Here $f_{\theta}(\theta)$ represents the probability density function of $\theta$ in $F$ and $F_{v|\theta}[\ell_p(\theta)|\theta]$ represents the conditional cumulative probability function. Thus we can simplify our problem as
\begin{align}
    \max& \int_{\theta}(\ell_p(\theta)-\theta\ell_p'(\theta))(1-F_{v|\theta}[\ell_p(\theta)|\theta])f_{\theta}(\theta)\mathrm{d}\theta
    \label{new_main}\\
   \text{s.t.} & ~\ell_p(\cdot) \textrm{ is concave and non-decreasing}.\nonumber
\end{align}
\cmt{
The following result can be seen as an evidence that posted-pricing is a good choice with a mild assumption in our daily life.
\begin{theorem}
If $F_{v|\theta}[v|\theta]=F_v[v]$, i.e., the distributions of $v$ and $\theta$ are independent, the optimal separation function $\ell_p(\cdot)$ is a constant function. 
\end{theorem} 
\begin{proof}
By Equation~\eqref{new_main}, if $F_{v|\theta}[v|\theta]=F_v[v]$, we have that
\begin{eqnarray*}
    &&\int_\theta (\ell_p(\theta)-\theta \ell_p'(\theta))(1-F_{v|\theta}[\ell_p(\theta)|\theta])f_{\theta}(\theta)\mathrm{d}\theta\\
    &=&\int_\theta (\ell_p(\theta)-\theta \ell_p'(\theta))(1-F_v[\ell_p(\theta)])f_{\theta}(\theta)\mathrm{d}\theta\\
    &\le&\int_\theta f_{\theta}(\theta)\max_{\ell_p(\theta)} \{\ell_p(\theta)(1-F_v[\ell_p(\theta)])\}\mathrm{d}\theta.
\end{eqnarray*}
We set $\ell_p(\theta)=\argmax_q q(1-F_v(q))$ for all $\theta$, i.e., the separation line is horizontal. The inequality will become an equality and $\ell_p(\cdot)$ is the optimal solution of Problem~\eqref{new_main}. 

Next, we infer the pricing function. By definition in Equation~\eqref{infer_p}, we have $H=\{0\}$,  
$p(0)=\ell_p(\theta)=\argmax_q q(1-F_v(q))$ and $p(t)=p(0)$ for all $t$. Hence, the pricing scheme is to set a posted price and each buyer will buy it directly if their value of the item is equal to or larger than the price and spend no extra time on the deal.
\end{proof}
}

\section{The Optimal Pricing Scheme}
Above we have shown that the posted-pricing is optimal if we assume that the distributions of the value and the CPUT are independent.
However, it is not the case when we consider arbitrary distributions, which makes the question much more complicated.
In this section, we first show how to find the optimal scheme in a discrete distribution setting and then shows that the problem with a continuous distribution can be solved with arbitrarily accuracy by applying discretization technique.

The following theorem shows that given a discrete distribution of buyers' types with finite size, the optimal pricing scheme can be found efficiently.
\begin{theorem}\label{thm:2}
Given a discrete distribution of buyers' types, whose support size is $n$, the optimal pricing scheme can be found in time complexity $O(n^4\log n)$.
\end{theorem}
\begin{proof}
Suppose these $n$ different types in the distribution are $\left\{Q_i=(\theta_i,v_i) \right\}_{i\in[n]}$. Without the loss of generality, we can assume that $\theta_1\le
\theta_2\le\ldots\le \theta_n$. In addition, there are $m\le n$ different
$\theta$'s over these $n$ types. We let $(\theta_{\beta(1)},\ldots,\theta_{\beta(m)})$ denote all the different CPUTs.

Given the optimal pricing scheme, each type chooses how much time to invest in the deal or leave without buying anything. We denote by $\ell^*$ the separation line in the optimal pricing.
Given any $i\in[m]$, we define a new function $K_i(\theta)={\ell^*}'(\theta_{\beta(i)}) (\theta - \theta_{\beta(i)}) + \ell^*( \theta_{\beta(i)})$.
Due to the uniqueness of $K_i$'s, we can construct a new separation function $K$ which outputs the minimum of
$K_i$'s given any $\theta>0$. Formally, we
define $$K(\theta)=\min_{i\in[m]}\left\{ K_i(\theta) \right\}.$$
It is easy to see that $K$ is a piecewise linear function with respect to $\theta$.
If we use $K$ as a separation line, each buyer with the same type will pay the same amount of money. Thus, $K$ is also the optimal solution.

We can then search over all such piecewise linear boundaries and find the
optimal one. Note that on any linear segment where two lines $K_i$ and $K$ overlap, i.e., $\{(\theta,K_i(\theta))|
K_i(\theta)=K(\theta)\}$, if it is not horizontal then there exist two types $Q_j$ and $Q_k$ lying on it. Otherwise, we can make
transformations as follows, maintaining the revenue non-decreasing:
\begin{itemize}
\item If the line segment
passes no type, we can push it upwards until it touches some discrete type. This
operation only increases the payment for some type.
\item If the line segment only passes one type $Q_j$ in the distribution, we rotate
  this line on $Q_j$ in the clock-wise direction until it touches another type, say $Q_k$ or becomes horizontal on the premise that this line passes through $Q_j$. 
  If $\theta_k<\theta_j$, type $Q_k$ will pay more. If $\theta_k>\theta_j$ and type $Q_k$ is already on the line $K$, then type $Q_k$'s payment is unchanged.
   If $\theta_k>\theta_j$ and type $Q_k$ is beneath the line $K$, then the buyer with type $Q_k$ will buy an item now. So this operation would never decrease the revenue.
\end{itemize}

We define two functions to make our descriptions easier. Let $d(i,j)=(v_j-v_i)/(\theta_j-\theta_i)$ denote the slope of the straight line segment connecting points $Q_i$ and $Q_j$. 
The other is the pricing function $z(i,j)=v_i-d(i,j)\theta_i$.
Let ${\rm cross}(i,j,k,h)$ denote the intersection point between the straight line passing types $Q_i,Q_j$ and the straight line passing $Q_k,Q_h$. 

We abuse the notation and let $f(Q_i,Q_j)$ (assume that $\theta_i<\theta_j$) denote the total sum of probabilities of types $Q=(\theta,v)$
lying above the straight line connecting type points  $Q_i$ and $Q_j$ and
$\theta\in[\theta_i,\theta_j)$. We define a partial order over the pair of types. Given $i$, $j$, $k$, $h\in[n]$, we say
$(i,j)\prec (k,h)$ if the point ${\rm cross}(i,j,k,h)$ is larger than
$Q_j=(\theta_j,v_j)$ and weakly smaller than $Q_k=(\theta_k,v_k)$ in both coordinates.

Now we are ready to present our dynamic programming algorithm. The recursive
formula is as follows,
\begin{align*}
  {\rm OPT}(i,j)=\max_{(i,j)\prec(k,h)}\{&z(i,j)f(Q_i,x) +{\rm OPT}(k,h)+ z(k,h)f(x,Q_k)\},
\end{align*}
where $x={\rm cross}(i,j,k,h)$.
Here ${\rm OPT}(i,j)$ represents the optimal revenue if the distribution is restricted to $(Q_i,Q_{i+1},\ldots,Q_n)$ and the line segment connecting $Q_i$ and $Q_j$ is part of the separation line.
To compute ${\rm OPT}(i,j)$, the algorithm enumerates all possible line segments that adjacent to the line segment connecting $Q_i$ and $Q_j$.
Then the optimal revenue is the maximum over $\max_{i,j} OPT(i,j)$. When we have found the optimal revenue, we can deduce the optimal separation function in reverse.

This dynamic programming can be solved efficiently. We first prepare all possible $w(Q_i,Q_j),~ \forall i,j$ values which takes $O(n^2)$ time. There are $O(n^2)$ sub-problems ${\rm OPT}(i,j)$ to solve. For each one, it enumerates all $(k,h)$ pairs which has $O(n^2)$ cases. Given $(i,j,k,h)$, $x$ can be solved in $O(1)$ time. By
utilizing the prepared values $w(Q_i,Q_j),~\forall i,j$, we can find $w(Q_i,x)$ in $O(\log n)$ time. Therefore, the overall complexity is $O(n^4\log n)$.
\end{proof}

Next we solve the continuous distributions approximately by discretization. The basic idea is clear: we partition the given joint continuous distribution as a discrete distribution, and calculate the optimal solution with respect to the discrete one efficiently. Given any continuous distribution $\mathcal{C}$ over  $[\underline{\theta},\overline{\theta}]\times[\underline{v},\overline{v}]$, we construct a new discrete distribution $\mathcal{D}$ as follows: we partition the space into squares of size $\epsilon\times\epsilon$ where $\epsilon>0$ is sufficiently small such that the probability of the types in each area $[k\epsilon,(k+1)\epsilon]\times[\underline{v},\overline{v}]$ is bounded by $\eta$. Then we put all the probability
density of one square to the mass point at the right bottom of the square. Note that $\eta$ could be arbitrarily small as long as $\epsilon$ is small enough. The following theorem says the approximation loss could be arbitrarily small.

\begin{theorem}\label{thm:3}
In the optimal scheme for distribution $\mathcal{C}$ and $\mathcal{D}$, the separation lines are denoted by ${\rm OPT}_C$ and ${\rm OPT}_D$ respectively. For the given distribution $\mathcal{C}$, the difference between the revenue achieved by separation lines ${\rm OPT}_D$ and ${\rm OPT}_C$ is within $\eta\overline{v}+\epsilon$.
\end{theorem}
\begin{proof}
Denote by $\widehat{\rm OPT}_C$ the separation function achieved by shifting ${\rm OPT}_C$ with a
distance $\epsilon$ downwards. This operation guarantees that every buyer who lies above the separation line ${\rm OPT}_C$ will stay above the separation line $\widehat{\rm OPT}_C$ after the discretization. Therefore, any buyer who purchases the item in the scheme corresponding to ${\rm OPT}_C$ will continue purchasing the item in the scheme corresponding to $\widehat{\rm OPT}_C$ but pay less. We use ${\rm OPT}_C(\mathcal{C})$ to represent the revenue achieved by ${\rm OPT}_C$ on distribution $\mathcal{C}$. $\widehat{\rm OPT}_C(\mathcal{D})$ is defined similarly. Then we have
\begin{equation}
 {\rm OPT}_C(\mathcal{C})\le \widehat{\rm OPT}_C(\mathcal{D})+\epsilon.\label{OPT_1}
\end{equation}

Since ${\rm OPT}_D$ is the optimal scheme for distribution $\mathcal{D}$, it holds that
\begin{equation}
 \widehat{\rm OPT}_C(\mathcal{D})\le {\rm OPT}_D(\mathcal{D}).\label{OPT_2}
\end{equation}

In the scheme corresponding to ${\rm OPT}_D$,
a buyer always purchases the item as long as her type after discretization purchases the item. But her payment might decrease due to the horizontal shift in the discretization. For a buyer whose type lying in the space $[k\epsilon,(k+1)\epsilon]\times[\underline{v},\overline{v}]$,
her payment decreases at most
\begin{align*}
  [\ell_D(\theta)-\theta\ell'_D(\theta)]\left|^{(k+1)\epsilon}_{\theta=k\epsilon}\right.,  
\end{align*}
then the total decrease of the payment is at most
\begin{equation*}
    \sum_k \eta\cdot [\ell_D(\theta)-\theta\ell'_D(\theta)]\left|^{\theta=(k+1)\epsilon}_{k\epsilon}=\eta(\overline{v}-0) \right..
\end{equation*}
Thus we have that 
\begin{equation}
{\rm OPT}_D(\mathcal{D})\le {\rm OPT}_D(\mathcal{C})+\eta\overline{v}.\label{OPT_3}
\end{equation}

The combination of Equations ~\eqref{OPT_1}, \eqref{OPT_2} and \eqref{OPT_3}, we have
\begin{equation*}
    {\rm OPT}_C(\mathcal{C})\le {\rm OPT}_D(\mathcal{C})+\eta\overline{v}+\epsilon. 
\qedhere
\end{equation*}
\end{proof}

\section{Revenue and Wasted Time}
In this section, we investigate the impact by switching from fixed pricing scheme to our optimal pricing scheme. Specifically, we first examine the increased revenue and then the degree of wasted time. 

In real-world scenarios, the seller may set the pricing function as a step function for simplicity. When the time spent approaches a predefined threshold, the price decreases to a lower level. As such, the separation function becomes piecewise constant. The following result demonstrates that a $k$-piecewise linear separation function could achieve $k$ times as much revenue as the optimal fixed pricing scheme. 

\begin{theorem}
For $k$-step pricing scheme, the revenue is at most $k$ times the optimal fixed pricing scheme. Moreover, the ratio is tight in the worst case.
\end{theorem}
\begin{proof}
Suppose there are $k$ different prices $p_1,\ldots,p_k$ in the pricing
function. The total revenue is the sum of payments collected from the buyers who
pay $p_i$ in this pricing scheme for all $i\in[k]$. Note that the revenue
gained from the buyers, who pay $p_i$ is no more than that of
the posted pricing with the price $p_i$, not to say the best scheme using fixed price.
Therefore, the total revenue using $k$-step pricing scheme is at most $k$ times the optimal fixed pricing scheme. It can be construed as the seller sells the item $k$ times and with a different fixed price in each time.

Now we illustrate an example and show that this above ratio is tight.
For any  $r>1$, $k$, and a small constant $\epsilon>0$, we construct a discrete
distribution with $k$ types $Q_i=(\theta_i,v_i)$, where $v_i=r^i$ and
$$
\theta_i = (r-1)\cdot \frac{(r/\epsilon)^i-1}{r/\epsilon-1}.
$$
The probability of being type $Q_i$ is as follows,
for $1\le i\le k-1$ we set $\Pr[Q=Q_i]=r^{1-i}-r^{-i}$ and $\Pr[Q=Q_k]=r^{1-k}$
otherwise. One can easily verify that such a distribution is well defined.

The marginal value distribution is similar to the equal revenue distribution
where the revenue is unchanged, no matter what the posted price is. In our
distribution, the revenue is always $r$ for every posted price $r^i$, given $i\in[k]$.

We construct a separation function which almost yields the revenue $kr$
indeed. Our separation function $\ell(\cdot)$ connects all the type points using
straight line segment:
$$
\ell(\theta)=
\begin{cases}
  r^i + \epsilon^i(\theta-\theta_i), &  \theta\in[\theta_i,\theta_{i+1}); \\
  r^k, &\theta\ge\theta_k.
\end{cases}
$$

It is easy to see that our function $\ell(\cdot)$ is convex and weakly increasing since its derivative is weakly decreasing and always non-negative. Hence function $\ell$ is a valid separation function.

According to Equation~\eqref{new_main} in continuous setting, the revenue under
the discrete distribution becomes
\begin{equation*}
\sum_{1\le i\le k}\left[ \ell(\theta_i)-\theta_i \ell'(\theta_i) \right]\cdot\Pr[Q=Q_i] .
\end{equation*}
By our assumption, when a buyer has several actions with the same highest
  utility, she will pays the highest possible money in favor of the seller.
Therefore, we should set $\ell'(\theta_i)$ to be the right-handed derivative.
\begin{eqnarray*}
    \ell(\theta_i)-\theta_i \ell'(\theta_i)    &=&r^i-\epsilon^i\theta_i\\
    &= &r^i-\frac{\epsilon^i(r-1)((r/\epsilon)^i-1)}{r/\epsilon-1}\\
    &\ge &r^i-\frac{r-1}{r-\epsilon}\cdot \epsilon r^i\\
    &\ge &(1-\epsilon)r^i.
\end{eqnarray*}
By plugging it into the revenue formula, the revenue is at least
\begin{equation*}
  \sum_{1\le i\le k}(1-\epsilon)r^i(r^{1-i}-r^{-i})=\frac{1-\epsilon}{1-r^{-1}}rk.
\end{equation*}
Taking the limit as $\epsilon$ to zero and $r$ to infinity yields the ratio to be $k$.
\end{proof}

People are wasting time in their pursuit of low prices. We measure the value of wasted time. Suppose a customer with type $(\theta,v)$ would buy one item given the separation function $\ell(\cdot)$. According to Lemma~\ref{lemma_time}, she spends $\ell'(\theta)$ amount of time. We name her value of this amount of time as time loss. In this case, her time loss is $l'(\theta) \theta$. Her money payment is $\ell(\theta)-\ell'(\theta)\theta$. Let $loss$ denote the total time loss from all customers and $rev$ denote the total money payment collected from all customers. The following theorem says while pursuing profit it could also cause a serious waste of time.

\begin{theorem}
Given a discrete distribution with $k$ different types, we have $loss\le (k-1)rev$ when using the optimal pricing scheme. The ratio is tight.
\end{theorem}

\begin{proof}
We consider the optimal separation line, denoted by $\ell(\cdot)$. 
According to the construction algorithm of the optimal separation line introduced in Theorem~\ref{thm:2}, the separation line would be horizontal when $\theta$ is large enough and there is a type with which the buyer does not have a time loss. We assume it is type $(\theta_k,v_k)$. We abuse the notation and let $f_i$ denote the probability of being type $(\theta_i,v_i)$.  For every type $(v_i,\theta_i)$, we could set a posted price and collect a revenue at least $v_if_i$.

Hence, we have $rev\ge v_i$. For every type $(v_i,\theta_i), i<k$, if $v_i\ge\ell(\theta_i)$, her loss would be $loss_i=(v_i-\ell'(\theta)\theta)f_i\le v_if_i\le rev$. Therefore, we have $loss\le \sum_{i<k} loss_i \le (k-1)rev$.

To prove the ratio is tight, we construct an instance. Let $d$ denote a large number which will be determined later. Let $\epsilon$ denote a small number close to zero which will be determined later.
We give a discrete distribution with $k$ possible types. 
For $i\le k$, we define $\theta_i=1+d+\cdots+d^{i-1}=\frac{d^i-1}{d-1}$ and $v_i=1+d(1-\epsilon)+d^2(1-2\epsilon)+\cdots+d^{i-1}(1-(i-1)\epsilon)=\frac{d^i-1}{d-1}-\epsilon(i-1)\cdot\frac{d^i}{d-1}+\epsilon d\cdot \frac{d^{i-1}-1}{(d-1)^2}$.
Define weight $w_i$ as
\begin{equation*}
    w_i=\prod\limits_{2\le j\le i}\frac{\theta_{j-1}\cdot d^{2-j}}{1+d+\theta_{j-1}\cdot d^{1-j}}.
\end{equation*}
Next we set $f_i=w_i/w$ where $w=\sum_{1\le j\le k}w_k$. It can be shown that $w_i=d^{1-i}(1+O(1/d))$.
Given this probability distribution with $k$ different types, it is easy to check that the optimal separation function is exactly the piecewise linear function that connects these points consecutively. 
Formally, the optimal function $\ell(\cdot)$ is
\begin{equation*}
    \ell(\theta)=\begin{cases}
       \frac{v_1}{\theta_1} \cdot \theta,& \theta\in [0, \theta_1];\\
      v_{i-1}+\frac{v_i-v_{i-1}}{\theta_i-\theta_{i-1}}\cdot (\theta-\theta_{i-1}), & \theta\in (\theta_{i-1},\theta_i];\\
      v_k, & \theta\in (\theta_k,\infty),
    \end{cases}
\end{equation*}
where $i=2,\ldots,k$.
Then we measure the extracted revenue and the time loss for each type. 
The revenue from type $(\theta_i,v_i), i<k$ is
\begin{eqnarray*}
 rev_i&=&(v_i-\theta_i\cdot (v_{i+1}-v_i)/(\theta_{i+1}-\theta_i))\cdot f_i\\
 &=&(v_i-\theta_i\cdot (1-i\epsilon))\cdot f_i\\
 &=&\epsilon\cdot \left(\frac{d^i-1}{d-1}+\frac{d^i-d}{(d-1)^2}\right)\cdot \frac{w_i}{w}\\
 &=&\epsilon/w \cdot(1+O(1/d)).
\end{eqnarray*}
The time loss from the same type is 
\begin{align*}
    loss_i &=\theta_i\cdot \frac{ v_{i+1}-v_i}{\theta_{i+1}-\theta_i}\cdot f_i\\
    &=\theta_i\cdot (1-i\epsilon) \cdot w_i/w\\
    &=d^{i-1}(1+O(1/d))(1-i\epsilon)\cdot d^{1-i}(1+O(1/d))/w\\
    &=(1+o(1))/w.
\end{align*}
The revenue collected from the type $(\theta_k,v_k)$ is 
\begin{eqnarray*}
rev_k&=&v_k f_k\\
&=&\frac{d^i-1}{d-1}(1+O(\epsilon k))\cdot d^{1-i}(1+O(1/d))/w\\
&=&(1+O(\epsilon n)+O(1/d))/w.
\end{eqnarray*}
There is no time loss for type $k$. To sum up, as long as $d\gg k$ and $\epsilon=o(1)$, we have $rev=\sum_{i\le k}rev_i=(1+o(1))/w$ and $loss=\sum_{i\le k-1}loss_i=(1+o(1))(k-1)/w$.
Then $loss/rev=k-1+o(1)$.
\end{proof}

\begin{figure}[ht]
    \centering 
\tikzset{every picture/.style={line width=0.75pt}} 
\resizebox{3.5cm}{4cm}{
\begin{tikzpicture}[x=0.75pt,y=0.75pt,yscale=-1,xscale=1]

\draw  (83.91,250.35) -- (236.8,250.35)(99.97,103.4) -- (99.97,272.2) (229.8,245.35) -- (236.8,250.35) -- (229.8,255.35) (94.97,110.4) -- (99.97,103.4) -- (104.97,110.4)  ;
\draw   (100,130) -- (180,130) -- (180,250) -- (100,250) -- cycle ;
\draw  [fill={rgb, 255:red, 228; green, 228; blue, 228 }  ,fill opacity=1 ] (180,130) -- (180,170) -- (100,250) -- (100,210) -- cycle ;
\draw    (180.27,197.64) -- (127.18,198.18) -- (100.2,223.2) ;
\draw  [dash pattern={on 4.5pt off 4.5pt}]  (99.91,198) -- (129.36,198) ;
\draw  [dash pattern={on 4.5pt off 4.5pt}]  (127.18,198.18) -- (127,250.18) ;

\draw (86.4,101.6) node [anchor=north west][inner sep=0.75pt]    {$v$};
\draw (220.6,255.4) node [anchor=north west][inner sep=0.75pt]    {$\theta $};
\draw (88,252.4) node [anchor=north west][inner sep=0.75pt]  [font=\footnotesize]  {$0$};
\draw (138,252.4) node [anchor=north west][inner sep=0.75pt]  [font=\footnotesize]  {$1$};
\draw (88,207.4) node [anchor=north west][inner sep=0.75pt]  [font=\footnotesize]  {$1$};
\draw (178,252.4) node [anchor=north west][inner sep=0.75pt]  [font=\footnotesize]  {$2$};
\draw (88,167.4) node [anchor=north west][inner sep=0.75pt]  [font=\footnotesize]  {$2$};
\draw (88,127.4) node [anchor=north west][inner sep=0.75pt]  [font=\footnotesize]  {$3$};
\draw (182.5,181.4) node [anchor=north west][inner sep=0.75pt]  [font=\footnotesize]  {$\ell ( \theta )$};
\draw (81,192.4) node [anchor=north west][inner sep=0.75pt]  [font=\tiny]  {$4/3$};
\draw (81,222.4) node [anchor=north west][inner sep=0.75pt]  [font=\tiny]  {$2/3$};
\draw (120,256.4) node [anchor=north west][inner sep=0.75pt]  [font=\tiny]  {$2/3$};

\end{tikzpicture}

}
    \caption{Buyers' type distribution}
    \label{fig:2}
\end{figure}
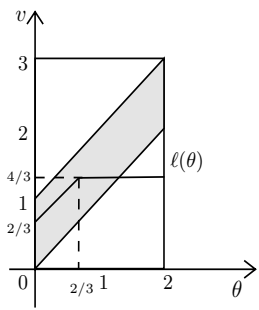

\section{Revenue and Correlations}
In this section, we show how to use calculus method to find optimal pricing scheme in an example and explain the limitations of this method. At last we discuss the implications from the example. 

We give a closed form of the optimal solution to a challenging case where the value and the CPUT are positively correlated shown in Figure~\ref{fig:2}. To be specific,
\begin{eqnarray*} 
    f(\theta,v)=\begin{cases}
        \frac{1}{2},  &   ~~~0 \le v - \theta \le 1\text{ and }0\le \theta \le 2;\\
        0,  & ~~~\text{otherwise}.
    \end{cases}
\end{eqnarray*}

\begin{theorem}
Consider the distribution defined by $f$ as above, the optimal separation line is  
\begin{eqnarray*} 
    \ell_p(\theta)=\begin{cases}
        2/3 + \theta,  &   ~~~0 \le \theta \le 2/3;\\
        4/3, & ~~~2/3 < \theta;\\
        0,  & ~~~\text{otherwise},
    \end{cases}
\end{eqnarray*}
and the optimal pricing scheme is 
\begin{eqnarray*} 
    p(t)=\begin{cases}
        4/3,  &   ~~~t\in [0,1];\\
        2/3, & ~~~t>1.
    \end{cases}
\end{eqnarray*}
\end{theorem}

\begin{proof}
According to the distribution, we have that $F_{v|\theta}[\ell_p(\theta)|\theta]=\theta+1-\ell_p(\theta)$ for $\ell_p(\theta)\in [\theta,\theta+1]$.
We can simplify the objective function in Equation~\eqref{new_main} as
\begin{eqnarray*}
    \frac12\int_{\theta=0}^2 (\ell_p(\theta)-\theta \ell_p'(\theta))[\theta+1-\ell_p(\theta)]^+ \mathrm{d}\theta,
\end{eqnarray*}
where we use $[\cdot]^+$ to denote the larger one between the input and zero. 
W.L.O.G., we claim that $\ell_p(0)$ is non-negative. Assuming that $\ell_p(0)< 0$, we can draw a shooting line $\ell_2(\cdot)$ from origin point $(0,0)$ touching the separation function $\ell_p(\cdot)$ with the touching point $(\hat{\theta},\ell_p(\hat{\theta}))$. 
If we use $\ell_p$ as the separation function, the payment from the buyer with type $\theta<\hat{\theta}$ will be negative. Therefore, the performance of the separation function 
$\hat{\ell}_p$ by replacing $\ell_p$ with $\ell_2$ over the interval $(0,\hat{\theta})$. 

Since $\ell_p(\theta)$ is concave, we define $b$ such that
$[0,b]=[0,2]\bigcap\{\theta|\ell_p(\theta)\ge \theta\}.$
In other words, $[0,b]$ is the interval where the separation function is above the line $v=\theta$.

At present, we focus on the payment from the buyer with $\theta\in [0,b]$ which depends on the design of the separation function $\ell_p(\theta),\theta\in [0,b]$. 
We also relax the term $[\theta+1-\ell_p(\theta)]^+$ to be $\theta+1-\ell_p(\theta)$ since it cannot be negative in the optimal solution. 

To find the optimal solution $\ell_p$, we use the method from the calculus of  variations. 
We define the Lagrangian $L$ as follows:
\begin{align*}
    L(\ell_p+\alpha g)=\int_{\theta=0}^b &\left[\ell_p(\theta)+\alpha g(\theta)-\theta \ell_p'(\theta)-\theta\alpha g'(\theta)\right] \cdot\left[\theta+1-\ell_p(\theta) - \alpha g(\theta)\right]\mathrm{d}\theta.
\end{align*}
    
When $\ell_p$ is the optimal solution, its restriction on $[0,b]$ is also the optimal solution with fixed value between $\ell_p(0)$ and $\ell_p(b)$. Then as long as $\ell_p+\alpha g$ satisfies convex and non-decreasing constraint and $g(0)=g(b)=0$, 
we have that $\frac{\partial L}{\partial \alpha}|_{\alpha=0}$ is zero.
\begin{align*} 
    \frac{\partial L}{\partial \alpha} =&\int_{0}^b \bigg[(g(\theta)- \theta g'(\theta))(- \ell_p(\theta)+\theta +1 )  - [\ell_p(\theta) - \theta \ell_p'(\theta)]g(\theta)  \bigg] \mathrm{d}\theta\\
    =&\int_{0}^b g(\theta)\big(\theta+1-2\ell_p(\theta)+\theta \ell_p'(\theta)\big)\mathrm{d}\theta - \int_{0}^b \theta(\theta +1 -\ell_p(\theta))g'(\theta)\mathrm{d}\theta\\
    =&\int_{0}^b g(\theta)\left(\theta+1-2\ell_p(\theta)+\theta \ell_p'(\theta)\right)\mathrm{d}\theta - \theta(\theta+1-\ell_p(\theta))g(\theta) \vert_0^b +\int_{0}^b g(\theta)\left( 2\theta+1-\ell_p(\theta)-\theta\ell_p'(\theta)\right) \mathrm{d}\theta\\
    =&\int_{0}^b (3\theta+2-3\ell_p(\theta))g(\theta)\mathrm{d}\theta.
\end{align*}
For any interval where $3\theta+2-3\ell_p(\theta)>0$, any solution $\ell_p(\theta)+\alpha g(\theta)$ is feasible and $\alpha>0$, we have $g(\theta)=0$. With the fact that $g(0) = g(\theta)=0$, the function $\theta(\theta+1-\ell_p(\theta))g(\theta) \vert_0^b$ can be dropped from the equation. It implies that in such an interval, $\ell_p(\theta)$ has already achieves the highest possible value. 
Similarly, in any interval where $3\theta+2-3\ell_p(\theta)\ge 0$ , $\ell_p(\theta)$ achieves the lowest possible value. 

Assume the intersection between $\ell_p$ and straight line $v=\theta+2/3$ is at point $(\tilde{\theta},\tilde{\theta}+2/3)$. According to the argument, we have $\ell_p$ is a straight line for $\theta\in [0,\tilde{\theta}]$ and
is a horizontal line for $\theta\in [\tilde{\theta},b]$.

By simple calculations, the optimal intersection point is $(\tilde{\theta}, \ell_p (\tilde{\theta}))=(2/3,3/4)$. The optimal separation function is given as follows,
\begin{eqnarray*} 
    \ell_p(\theta)=\begin{cases}
        2/3 + \theta,  &   ~~~0 \le \theta \le 2/3;\\
        4/3, & ~~~2/3 < \theta;\\
        0,  & ~~~\text{otherwise}.
    \end{cases}
\end{eqnarray*}
The pricing scheme can be computed easily and the calculation is omitted.
\end{proof}

As shown in the proof, we can see that the calculation of the closed-form of the separation line heavily depends on the distribution assumption, which could be much more complicated if the distribution is not very well.

When the value $v$ and the CPUT $\theta$ are positively correlated, the pricing scheme may help, as illustrated in Theorem~\ref{thm:3}. In the positive correlation example, the optimal revenue using posted price is $25/32$, and the optimal revenue using our pricing scheme is $22/27$, which increases about $\frac{22/27}{25/32}-1\approx 4.6\%$. 
The time loss is $1/27$, which is slightly smaller than the increase of the revenue.

\section{Conclusion}
In this paper, we introduce the pricing scheme problem with time-sensitive buyers. We then propose a polynomial algorithm to find the optimal scheme in the discrete distribution setting. We consider the revenue change and measure wasted time in our pricing scheme. There are several future directions based on this work. In this paper, we study the revenue maximization problem. It could be interesting to know some results about the welfare maximization. Now we focus on the offline model in the sense that the customers are known in advance. We are also interested this problem in the online model. Furthermore, we assume that the cost of time is linear with respect to the time, one interesting question is to explore general cost functions which could be more practical.

\bibliographystyle{IEEEtran}  
\bibliography{refs}

\end{document}